\documentclass[11pt]{article}

\usepackage[numbers]{natbib}
\usepackage{xspace}
\usepackage[usenames,dvipsnames,table]{xcolor}
\usepackage{graphicx}
\usepackage{amsmath,amssymb,amsthm}
\usepackage[margin=2.55cm]{geometry}
\usepackage{pstricks}
\usepackage{comment}
\usepackage{enumerate}
\usepackage{caption}

\usepackage{thmtools}
\usepackage{thm-restate}

\usepackage{fourier}

\usepackage{cleveref}

\usepackage[pdfstartview=FitH,colorlinks=true,citecolor=blue,linkcolor=blue]{hyperref}

\newcommand{\figref}[1]{Figure~\ref{#1}}

\newcommand{\ie}{i.e.,\xspace}

\newcommand{\gain}[2]{\textsc{G}_{#1}\!\left(#2\right)}
\newcommand{\core}{{\normalfont\textsc{Core}}\xspace}
\newcommand{\correct}{{\normalfont\textsc{Correction}}\xspace}

\newcommand{\allocation}{\mathcal{A}}
\newcommand{\tallocation}{\mathcal{X}}
\newcommand{\residue}{R}

\newcommand{\mysetminusD}{\hbox{\tikz{\draw[line width=0.6pt,line cap=round] (3pt,0) -- (0,6pt);}}}
\newcommand{\mysetminusT}{\mysetminusD}
\newcommand{\mysetminusS}{\hbox{\tikz{\draw[line width=0.45pt,line cap=round] (2pt,0) -- (0,4pt);}}}
\newcommand{\mysetminusSS}{\hbox{\tikz{\draw[line width=0.4pt,line cap=round] (1.5pt,0) -- (0,3pt);}}}
\newcommand{\mysetminus}{\mathbin{\mathchoice{\mysetminusD}{\mysetminusT}{\mysetminusS}{\mysetminusSS}}}

\DeclareMathOperator*{\argmax}{arg\,max}

\theoremstyle{definition}
\newtheorem{definition}{Definition}

\newtheorem{remark}{Remark}
\usepackage{subcaption}
\usepackage{enumitem}
\usepackage{thmtools}
\usepackage{thm-restate}

\newenvironment{proofof}[1]
{\medskip\noindent{\it #1.}\hspace{1ex}}

\usepackage{ifthen}
\newenvironment{rtheorem}[3][]{%
	\noindent\ifthenelse{\equal{#1}{}}{\bf #2 #3.}{\bf #2 #3 (#1)}%
	\begin{it}}{\end{it}}

\theoremstyle{plain}
\newtheorem{theorem}{Theorem}

\newtheorem{lemma}{Lemma}
\newtheorem{corollary}{Corollary}

\usepackage[linesnumbered,ruled,vlined]{algorithm2e}
\usepackage{etoolbox}
\newcounter{algoline}
\newcommand\Numberline{\refstepcounter{algoline}\nlset{\thealgoline}}
\AtBeginEnvironment{algorithm}{\setcounter{algoline}{0}}

\usepackage{tikz}
\usetikzlibrary{positioning,arrows.meta,arrows,decorations.markings}
\tikzset{main node/.style={circle,fill=blue!20,draw,minimum size=1cm,inner sep=0pt},}
\usetikzlibrary{calc,decorations.pathreplacing}

\begin{document}
\title{An Improved Envy-Free Cake Cutting Protocol for Four Agents}
\date{}
\author{Georgios Amanatidis\thanks{Centrum Wiskunde \& Informatica (CWI), Netherlands.} \and George Christodoulou\thanks{University of Liverpool, UK.} \and John Fearnley\footnotemark[2] \and Evangelos Markakis\thanks{Athens University of Economics and Business, Greece, Contact author.} \and Christos-Alexandros Psomas\thanks{Carnegie Mellon University, USA.} \and Eftychia Vakaliou\footnotemark[3]}	
\maketitle

\maketitle

\begin{abstract}
We consider the classic cake-cutting problem of producing envy-free allocations, restricted to the case of four agents. The problem asks for a partition of the cake to four agents, so that every agent finds her piece at least as valuable as every other agent's piece. The problem has had an interesting history so far. Although the case of three agents is solvable with less than 15 queries, for four agents no bounded procedure was known until the recent breakthroughs of 
Aziz and Mackenzie~\cite{AM16a,AM16b}. The main drawback of these new algorithms, however, is that they are quite complicated and with a very high query complexity. With four agents, the number of queries required is close to 600. In this work we provide an improved algorithm for four agents, which reduces the current complexity by a factor of 3.4. Our algorithm builds on the approach of 
\cite{AM16a} by incorporating new insights and simplifying several steps.  
Overall, this yields an easier to grasp procedure with lower complexity.
\end{abstract}

\section{Introduction}
Producing an envy-free allocation of an infinitely divisible resource is a classic problem in fair division.
As it is customary in the literature, the resource is represented by the
interval $[0,1]$, and each agent has a probability measure encoding her preferences over subsets of $[0, 1]$.
The goal  is to divide the entire interval among the agents so that no one envies the subset received by another agent. We note that the partition does not need to consist of contiguous pieces; the piece of an agent may be any finite collection of subintervals. 

The problem has a long and intriguing history. It has been long known that envy-free allocations exist for any number of agents, via non-constructive proofs \cite{DB61,Stromquist80,su1999rental}. 
For algorithms, the standard approach is to assume access to the
valuation functions via \textit{evaluation} and \textit{cut queries} (see Section
\ref{sec:prelims}). Under this model, we are interested in counting the number
of queries needed for producing an envy-free allocation. For two agents, the
famous cut-and-choose protocol requires only two queries. For three agents, the
procedure of Selfridge and Conway \cite{BT96}  guarantees an envy-free
allocation after at most 14 queries. For four agents and onwards, however, the picture changes drastically. The first finite, yet unbounded
, algorithm was proposed by \cite{BT95}. This was followed up by other more intuitive algorithms, which are also unbounded, e.g., \cite{RW97,Pikhurko00}.
Finding a bounded algorithm was open for decades and positive results had been known only for certain special cases, like piece-wise uniform or polynomial valuations \citep{AY14,KLP13,Branzei15}.
It was only recently that a major breakthrough was achieved by Aziz and
Mackenzie, presenting the first bounded algorithms, initially for four agents
\citep{AM16a}, and later for an arbitrary number of agents~\citep{AM16b}.

Despite these significant advances, the algorithms of \cite{AM16a,AM16b} are still of very high complexity. For an arbitrary number of agents, $n$, the currently known upper bound involves a tower of exponents of $n$, and even for the case of four agents, the known algorithm requires close to $600$ queries. On top of that, these algorithms are rather complicated and their proof of correctness requires tedious case analysis in certain steps. Hence, a clean-cut and more intuitive algorithm is still missing. 
\medskip

\noindent {\em Contribution:} 
We focus on the case of four agents and present an improved algorithm that reduces the query complexity roughly by a factor of $3.4$ (requiring $61$ cut queries and $110$ evaluation queries). 
Our algorithm utilizes  building blocks that are similar to the ones used
by~\cite{AM16a}, but by incorporating new insights and simplifying several
steps, we obtain a solution with significantly fewer
queries. The main differences between our work and \cite{AM16a} are highlighted at the end of this section. 
Our algorithm works by maintaining a partial allocation along with a leftover residue. Throughout its execution, it keeps updating the allocation and reducing the residue, until certain structural properties are satisfied. These properties involve the notion of {\em domination}, where we say that an agent $i$ dominates another agent $j$, if allocating the whole remaining residue to $j$ will not create any envy for $i$. A crucial part of the algorithm is to get a partial allocation where one agent is dominated by two others.   
Once we establish this, we then exhibit how to produce a complete allocation of the cake without introducing any envy. Overall, this results in an algorithm with markedly lower query complexity.
\medskip

\noindent {\em Further related work:}
We refer the reader to the book chapters \cite{LR15,Procaccia16-survey} for a more proper treatment of the related literature. Towards simplifying the algorithm of Aziz and Mackenzie \cite{AM16a}, the work of Segal-Halevi et al.~\cite{HHA15} (see their Appendix B) proposes a conceptually simpler framework, 
without, however,  improving the query complexity. Apart from the algorithmic results mentioned above, there has also been a line of work on lower bounds.
For envy-freeness, Stromquist \cite{Stromquist08} showed that there is no finite protocol for producing envy-free allocations where all the pieces are contiguous. 
Later on, Procaccia \cite{Procaccia09} established an $\Omega(n^2)$ lower bound
for producing non-contiguous envy-free allocations. Apparently, there is still a
huge gap between the known lower and  upper bounds for the problem for any $n\geq 4$. 
Interestingly, for the stricter notion of {\em strong envy-freeness}, where we require each agent to believe she is strictly better off than anyone else, the known lower bound is also $\Omega(n^2)$ \cite{BKM05}.

\subsubsection*{An Overview of the Algorithm}

We start with a high level description of the main ideas. 
As with most other algorithms, our algorithm maintains throughout its execution a partial allocation of the cake, along with an unallocated residue. 
The goal is to keep updating the allocation and diminishing the residue, with the invariant that the current partial allocation is always envy-free. Once the residue is eliminated, we are left with a complete envy-free allocation.
As mentioned earlier, the notion of {\em domination} is pivotal in our approach. The algorithm creates certain domination patterns between the agents, working in phases as follows:\medskip

\noindent {\em Phase One.}
We find  this first phase of particular importance, as it is also the most computationally demanding one. 
Here the goal is to get a partial envy-free allocation in which some agent is dominated by two other agents as in Figure \ref{fig:desired_graph}.
In order to establish dominations among agents, we use as a subroutine the so-called $\core$ protocol. 
In the $\core$ protocol one agent has the role of the ``cutter'', and the output is a new allocation with a strictly diminished residue. 
The properties of \core have several interesting and crucial consequences. First, if $\core$ is executed twice with the same agent as the cutter, then this cutter dominates at least one other agent in the resulting allocation. Moreover, if we run \core two more times, we may not get any extra dominations right away but we can still make a small {\em correction} so that the cutter dominates one more agent. This is done by using a protocol, referred to as the \correct protocol, which performs a careful redistribution.
Finally, by running \core one more time with a different cutter and the current residue, we show how further dominations arise that lead to the desired structure of one agent being dominated by two others. In total, phase one requires up to 6 calls to the \core protocol. \medskip

\noindent {\em Phase Two.} Suppose that at the end of phase one, agent $A$ is dominated by agents $B$ and $C$. 
The goal in the second phase is to produce a partial envy-free allocation where both $A$ and $D$ dominate both $B$ and $C$.
To achieve this goal, we execute $\core$ twice on the residue with $D$ as the cutter.
Then, if we still do not have the required dominations, we use again the $\correct$ protocol to  appropriately reallocate one of the last two partial allocations produced by $\core$. This suffices to create the dominations shown in Figure \ref{fig:phase_two}. \medskip

\noindent {\em Phase Three.}
Since both $B$ and $C$ are now dominated by $A$ and $D$, we can simply execute the cut-and-choose protocol for $B$ and $C$ on the remaining residue. 

\subsubsection*{Similarities and Differences with the Aziz-Mackenzie Algorithm in~\cite{AM16a}}
Our algorithm uses similar building blocks as the algorithm for four agents in \cite{AM16a}, combined with new insights. Namely, our $\core$  and  $\correct$ protocols  on a high level serve the same purpose as the core and the permutation protocols  in \citep{AM16a}. Conceptually, a crucial difference is the target structure of the domination graph. The initial (and most query-demanding) step of~\cite{AM16a} is to have {\it every} agent dominate two other agents. Here, our goal in phase one is to \emph{have just one agent dominated by two other agents}. Once this is accomplished, it is possible to reach a complete envy-free allocation much faster. Another important difference is the implementation of the $\core$ protocol itself. Our version is simpler regarding both its statement and its analysis. 
It also differs in the sense that it takes as input more information than in \cite{AM16a}, such as the current allocation, and it is not required to always output a partial envy-free allocation of the current residue. This extra flexibility allows us to avoid the tedious case analysis stated in the core protocol of \cite{AM16a} 
and, at the same time, further reduce the number of queries.

\hspace{8pt}
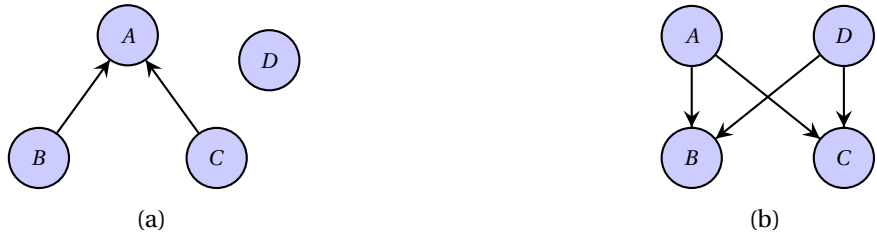
\begin{figure}[ht!]
	\centering
	\subcaptionbox{\label{fig:desired_graph}}[.49\linewidth]
	{
		\begin{tikzpicture}[thick,scale=0.8, every node/.style={scale=0.8}]
		\node[main node] (1) {$A$};
		\node[main node] (2) [below left = 1.05cm and 0.6cm of 1]  {$B$};
		\node[main node] (3) [below right = 1.05cm and 0.6cm of 1] {$C$};
		\node[main node] (4) [below right = -0.25cm and 1.3cm of 1] {$D$};
		
		\draw [decoration={markings,mark=at position 1 with {\arrow[scale=1.5,>=stealth]{>}}},postaction={decorate}] (2) to (1);
		\draw [decoration={markings,mark=at position 1 with {\arrow[scale=1.5,>=stealth]{>}}},postaction={decorate}] (3) to (1);
		\end{tikzpicture}
	}
	\subcaptionbox{\label{fig:phase_two}}[.49\linewidth]
	{
		\begin{tikzpicture}[thick,scale=0.8, every node/.style={scale=0.8}]
		\node[main node] (1) {$A$};
		\node[main node] (2) [right = 1.2cm of 1] {$D$};
		\node[main node] (3) [below = 0.8cm of 1] {$B$};
		\node[main node] (4) [below = 0.8cm of 2] {$C$};
		
		\draw [decoration={markings,mark=at position 1 with {\arrow[scale=1.5,>=stealth]{>}}},postaction={decorate}] (1) to (3);
		\draw [decoration={markings,mark=at position 1 with {\arrow[scale=1.5,>=stealth]{>}}},postaction={decorate}] (1) to (4);
		\draw [decoration={markings,mark=at position 1 with {\arrow[scale=1.5,>=stealth]{>}}},postaction={decorate}] (2) to (3);
		\draw [decoration={markings,mark=at position 1 with {\arrow[scale=1.5,>=stealth]{>}}},postaction={decorate}] (2) to (4);
		\end{tikzpicture}
	}
	\caption{An illustration of the domination graphs we want to achieve at the end of the first (a) and the second (b) phase respectively. In both graphs additional edges may be present but are not relevant.}\label{fig:domination}
\end{figure}

\section{Preliminaries}
\label{sec:prelims}
Let $N = \{ 1, 2, 3, 4 \}$ be a set of four agents. The cake is represented as the interval $\left[ 0, 1 \right]$; a \textit{piece} of the cake can be any finite union of disjoint intervals.  Each agent $i \in N$ is associated with a valuation function $v_i$ 
defined on all finite unions of intervals. We assume that the valuation functions satisfy the following standard properties for all $i \in N$: 
\vspace{-3pt}\begin{itemize}
	\item Normalization: $v_i\left( \left[ 0, 1 \right] \right) = 1$.
	\item Additivity: for all disjoint $X, X' \subseteq \left[ 0, 1 \right]$: $v_i \left( X \cup X' \right) = v_i \left( X \right) + v_i \left( X' \right)$.
	\item Divisibility: for every $[x, y] \subseteq \left[ 0, 1 \right]$ and every $\lambda\in [0, 1]$, there exists $z \in [x, y]$ such that $v_i \left([x, z] \right) = \lambda v_i \left( [x, y] \right)$. 
	Note that this implies that $v_i([x,x]) = 0$, for all $x\in [0, 1]$.
	\item Nonnegativity: for every $X \subseteq \left[ 0, 1 \right]$ it holds that $v_i \left( X \right) \geq 0$. 
\end{itemize}
\vspace{-3pt} By $\mathcal{X} = (X_1, X_2, X_3, X_4)$ we denote the allocation where  agent $i$ is given the piece $X_i$. 

\begin{definition}[Envy-freeness]\label{def:envy-freeness}
	An allocation $\mathcal{X} = (X_1, X_2, X_3, X_4)$ is \emph{envy-free}, if $v_i \left( X_i \right) \geq v_i \left( X_j \right)$, for all  $i,j\in N$, i.e., every agent prefers her piece to any other agent's piece. 
\end{definition}

We say that $\mathcal{X}$ is a \textit{partial} allocation, if there is some  cake that has not  been allocated yet, i.e., $\bigcup_{i=1}^4 X_i\subsetneq [0, 1]$. The unallocated cake is called the  \textit{residue}. 
During the execution of the algorithm the residue  diminishes, until eventually it becomes the empty set.
As we noted, an important notion is that of \textit{domination} or \textit{irrevocable advantage} \cite{BT96}. It will be insightful to think of a graph-theoretic representation of our goals, via the {\em domination graph} of the current allocation.

\begin{definition}[Domination and Domination Graph]
	Given a partial allocation $\mathcal{X} = (X_1, X_2, \allowbreak X_3, \allowbreak X_4)$ and a residue $R$, we say that an agent $i$ \emph{dominates} another agent $j$, if $v_i(X_i)\geq v_i(X_j \cup R)$. That is, 
	$i$ would not be envious of $j$ even if $j$ were allocated all of $R$. 
	The \emph{domination graph} with respect to $\mathcal{X}$ is a directed graph where the nodes correspond to the agents and there exists a directed edge $(i,j)$ if and only if agent $i$ dominates agent $j$.
\end{definition}

%

Achieving certain patterns in the domination graph can make the allocation of the remaining residue straightforward.
For example, if there exists a node $i$ with in-degree $3$, allocating all of the residue to agent $i$ results in an envy-free allocation. 
As another example, the protocol of~\cite{AM16a} tries to get a domination graph where every node has out-degree at least $2$. 
In our algorithm, we also enforce a certain structure on the domination graph. 


\subsubsection*{The Robertson-Webb Model}
The standard model in which we measure the complexity of cake cutting algorithms is the one suggested by Robertson and Webb \cite{RW98} and formalized by Woeginger and Sgall \cite{WS07}.
In this model, two kinds of queries are allowed:
\begin{itemize}
	\item \textit{Cut queries:} given an agent $i$, a point $x \in [0,1]$ and a value $r$, with $r\leq v_i \left( [x,1] \right)$,
	the query returns the smallest $y \in [0,1]$ such that $v_i \left( [x,y] \right) = r$.
	\item \textit{Evaluation queries:} given an agent $i$ and an interval $[x,y]$, return $v_i \left( [x,y] \right)$. 
\end{itemize}
Virtually all known discrete cake-cutting protocols can be analyzed within this framework.
For example, the cut-and-choose protocol is implemented as follows: the algorithm makes one cut query for agent 1 with $r=1/2$, starting from $x=0$. This is followed by an evaluation query on agent $2$ for one of the pieces (which also reveals the value of the second piece).

\subsubsection*{Conventions on Ties, Marks, Partial Pieces, and Residues} 
All algorithms in this work ignore ties. However, assuming an appropriate tie-breaking scheme, this is without loss of generality (also see the discussion in \cite{AM16b}).

We follow some conventions---also adopted in related work---when it comes to handling trims and partial pieces.
In various steps during the algorithm, one agent cuts the residue into pieces, and the other agents are asked to place marks on certain pieces. We always assume that marks are placed starting from the left endpoint of a piece, and this operation creates a partial piece, contained between the mark and the right endpoint.
In particular, suppose we have a partition of the residue into four contiguous pieces. Then, an 
agent may be asked to place a mark on her most favorite piece so that the resulting partial piece
has the same value as her second favorite piece (see \figref{fig:2-mark}).    
The types of marks that the algorithm needs are described in the following definition.

\vspace{5pt}
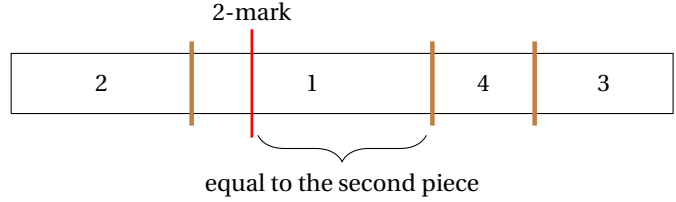
\begin{figure}[htb!]
	\centering
	\begin{tikzpicture}[scale=0.8, every node/.style={scale=0.9}]
	\draw (0,0) -- (11,0) -- (11,1.0) -- (0,1.0) -- (0,0);
	\draw[brown, line width=0.6mm] (3,1.2) -- (3,-0.2);
	\draw[brown, line width=0.6mm] (7,1.2) -- (7,-0.2);
	\draw[brown, line width=0.6mm] (8.7,1.2) -- (8.7,-0.2);
	
	\node[] at (4, 1.7)   (a) {$2$-mark};
	\draw[red, line width=0.35mm] (4,1.4) -- (4,-0.4);
	
	
	\draw[decorate,decoration={brace,amplitude=10pt,raise=3pt,mirror},yshift=2pt] (4.1,-0.3) -- (6.9,-0.3) node [midway,yshift=-24pt]{equal to the second piece};

	\node[] at (1.5, 0.5)   (a) {2};
	\node[] at (5, 0.5)   (a) {1};
	\node[] at (7.85, 0.5)   (a) {4};
	\node[] at (9.85, 0.5)   (a) {3};
	
	\end{tikzpicture}
	\caption{The view of the residue for a non-cutter at the time she performs a $2$-mark. }\label{fig:2-mark}
\end{figure}

\begin{definition}
	Given a partition  of the residue into four pieces, we say that an agent performs an $x$-mark, if she places a mark on each of her $x-1$ most valuable pieces so that the resulting partial pieces all  have the same value as her $x$-th favorite piece.
\end{definition}

In the description of the algorithm we  use $2$-marks and $3$-marks. Of course, after all marks are placed, each connected piece may have multiple marks on it. Whenever a connected piece $p$ is only partially allocated, the part $p'$ of $p$ that is allocated is always the interval between the second rightmost mark and the right endpoint of $p$. While at this point it is not clear whether a \textit{second} mark on a piece even exists, we will argue later on that marked pieces will have at least two marks (Lemma~\ref{lem:marked}). Hence, if some agent $i$ receives a partial piece $p'$, resulting from an initial piece $p$, it is not necessarily true that  $p'$ is defined by $i$'s own mark. However, in such a case the algorithm always makes sure that $i$ receives a part of $p$ that is beyond $i$'s own mark. Formally, we say that  $i$ is allocated a part of piece $p = [x,y]$ \textit{beyond}  (resp.~\textit{strictly beyond})  a mark $m$, 
if $i$ is allocated $[m',y]$ with $m' \le m$ (resp.~ $m' < m$).

Note that in the above discussion the residue is seen as a single interval, while in fact it may be a finite union of intervals. We keep this view throughout this work as it is conceptually easier and allows for a cleaner presentation. 
Asking queries on pieces can, of course, be simulated by asking queries on
intervals, but the number of the latter can grow linearly on the number of
intervals that make up a piece.\footnote{This issue has not been addressed in
	the query counting of \cite{AM16a}, but there the main goal was to obtain a
	bounded algorithm. Here we do keep track of the extra queries.} We take care of this by making sure that at any time, the algorithm knows for every agent the values of all the intervals that make up the residue (see the query counting argument in the last part of Section \ref{sec:core}). 

\section{The Algorithm}\label{sec:algorithm}

The main result of our work is the following.

\begin{theorem}\label{thm:main_protocol}
	The \textsc{Main Protocol} returns an envy-free allocation and makes at most $61$ cut queries, and $110$ evaluation queries.
\end{theorem}

We discuss first the main steps of our algorithm and  provide the relevant definitions and key properties, needed for the proof of correctness, in Section~\ref{sec:correctness}.

\paragraph*{Phase One.}
This is the most important part of the protocol, and computationally the most demanding one. The goal in phase one is to get a partial envy-free allocation, where some agent  is dominated by two other agents, i.e., the underlying domination graph has a node with in-degree at least $2$, as depicted in Figure~\ref{fig:desired_graph}.
In order to establish dominations among agents, we use a subroutine called  $\core$ protocol (stated in Section~\ref{sec:core}). 
This protocol takes as input a specified agent, called the {\em cutter}, the current partial allocation, and the current residue. 
For technical convenience, \core also takes as an  input a subset of  agents that we choose to \emph{exclude from competition} (this is made precise in the description of  \core  in Section~\ref{sec:core} but it roughly means that the excluded agents will choose their piece late in the \core protocol). In most cases, this argument is just the empty set. In particular, when no such argument is specified we mean that it is $\emptyset$. The output of $\core$ is a partial (\emph{usually} envy-free)
allocation of the residue  with some additional properties described below. 
In the initial step of $\core$, the cutter divides the current residue into four
equal-valued pieces according to her own valuation function.  Throughout the
protocol the rest of the agents---the \textit{non-cutters}---may mark these
pieces, and at the end, agents may be allocated either partial (marked) or complete pieces. 
Of course, if at any point $\core$ outputs an envy-free allocation of the whole cake, 
the algorithm terminates. 
The full description and the analysis of \core is given in Section~\ref{sec:core}. For now, we treat it as a black box 
and we assume that it satisfies the following properties. 

\begin{restatable}{cprop}{corepropone}\label{cproperty:whole-pieces}
	The cutter  and at least one more agent receive complete pieces, each worth exactly  $1/4$ of the value of the current residue 
	according to the cutter's valuation.
\end{restatable}

\vspace{-5pt}
\begin{restatable}{cprop}{coreproptwo}\label{cproperty:envy-free}
	The allocation output by any single execution of $\core$ when no agent is excluded from competition, is a (possibly partial) envy-free allocation.
\end{restatable}

The above properties allow us to deduce an important fact: if $\core$ is executed at least twice with the same agent as the cutter, then this cutter dominates at least one agent in the resulting allocation. In fact, we can be more specific about the agent who gets dominated. The important observation here, stated in Lemma~\ref{lem:core_domination} (Section \ref{sec:correctness}), 
is that a second run of $\core$ makes the cutter dominate whoever received the so-called {\em insignificant piece} in the first execution. 
\vspace{-1mm}
\begin{definition}
	Let $\allocation$ be an allocation produced by a single run of $\core$.  
	Among the four pieces given to the agents, the partial piece that is least desirable to the cutter is called the {\em insignificant piece} of $\allocation$. 
\end{definition}

Hence, if we run \core twice, say with agent $1$ as the cutter, we enforce one edge in the domination graph. In order to proceed further and obtain a node with in-degree two, we first attempt, as an intermediate step, to have a domination graph where one node has out-degree equal to two. One may think that the intermediate step can be achieved by running \core more times  with agent $1$ as the cutter.
The  problem with this approach is that even if we further execute $\core$ \emph{any} number of times, there is no guarantee that new dominations will appear; the same agent may receive the insignificant piece in every iteration.

To fix this issue, it suffices to run \core 4 times with agent $1$ as the cutter and then make a small \textit{correction} to one of the $4$ partial allocations produced by \core. 
In particular, denote by $\allocation^k = \{ p^k_1, p^k_2, p^k_3, p^k_4 \}$, with $k=1,...,4$, the suballocation output by the $k$th execution of $\core$ within the \textit{for} loop of line \ref{line:for_loop} of \textsc{Main Protocol}, and let $\residue^k$ be the residue after the $k$th execution. 
Then clearly for each agent $i$, $p_i^k \subseteq R^{k-1}$, and the current allocation of the algorithm after the 4 calls to \core is $\tallocation = \{p_1, p_2, p_3, p_4\}$, with $p_i = \bigcup_{k=1}^4 p^k_i$. 
Among these $4$ suballocations that $\tallocation$ consists of, we identify one in which we can perform a certain redistribution without introducing any envy. 
To do this, we exploit the notion of {\em gain}, which is the difference between the value that an agent has for her own piece compared to the pieces of agents she does not dominate. 

\begin{definition}[Gain]\label{def:gain}
	Let $\tallocation = \{p_1, \dots, p_4\}$ be the current partial allocation of the cake, and  $\allocation = \{p'_1, \dots, p'_4\}$ be a suballocation of  $\tallocation$, \ie $p'_i\subseteq p_i$ for $i\in N$. Further, let $D_i$ be the set of agents that are dominated by $i$ in $\tallocation$ and $N_i=N \mysetminus (D_i \cup \{i\})$. Then 
	the gain of $i$ with respect to  $\allocation$,
	$\gain{\allocation}{i}$, is the difference between $v_i(p'_i)$ and the maximum value of $i$ for a piece  in $\allocation$ given to any agent in $N_i$, 
	i.e., $\gain{\allocation}{i} = v_i(p'_i) - \max_{j \in N_i} v_i(p'_j)$.\footnote{Note that $\gain{\allocation}{i}$ is not defined when $N_i=\emptyset$. In fact, we never need it in such a case.}  
\end{definition}

\vspace{5pt}
{ \LinesNotNumbered
	\begin{algorithm}[H]
		\DontPrintSemicolon 
		\NoCaptionOfAlgo
		{\small
			\textbf{Phase One} \; 
			\Numberline \For{$count = 1$ to $4$\label{line:for_loop}}
			{\Numberline Run $\core$ on the current residue with agent $1$ as the cutter.\label{line:for_loop_agent1}} 
			\Numberline  \If{the same agent got the insignificant piece in all 4 executions of \core\label{line:main:done with first pass}\label{line:main:agent1_dom2}}
			{ \Numberline Find $\allocation^* \in \{ \allocation^1, \allocation^2, \allocation^3, \allocation^{4}\}$ such that $\gain{\allocation^*}{i} \leq \sum_{ \allocation \neq \allocation^* } \gain{\allocation}{i}$ for all $i\in N \mysetminus \{1\}$. \label{line:main:find_allocation} \;
				\Numberline Run $\correct$ on $\allocation^*$.\label{line:main:correct1} \;}
			\Numberline Run $\core$ on the residue with agent $1$ as the cutter.\label{line:main:core_again}\;
			\Numberline \If{there is some agent $E\in N \mysetminus \{1\}$ not dominated by agent $1$\label{line:main:done_w_double_dom}}
			{ \Numberline Run $\core$ on the residue with agent $E$ as the cutter,  excluding agent $1$ from competition.\label{line:main:finish_phase1}\;}
			\Numberline\label{line:main:else:S-C}\Else{\Numberline Run the Selfridge-Conway Protocol on the  residue for agents $2$, $3$, and $4$, and terminate.\label{line:main:S-C}\;}
			\textit{Now, if the algorithm has not terminated, some agent $A$ is dominated by two other agents  $B$ and $C$. 
				Let $D$ be the remaining agent.} \;
			\textbf{Phase Two} \;
			\Numberline \For{$count = 1$ to $2$}
			{ 
				\Numberline Run $\core$ on the current residue with agent $D$ as the cutter,  excluding from competition any one from $\{B, C\}$ who dominates two non-cutters. \; \label{line:main:core_phase2}}
			\Numberline \If{$B$ and $C$ are not both dominated by $A$ and $D$\label{line:main:B xor C}}
			{ \Numberline Let $F\in\{B, C\}$ be the agent who got the insignificant piece in the last two calls of $\core$. \;
				\Numberline Run $\correct$ on the suballocation (out of the last two) where $\gain{\allocation}{F}$ was smaller.}
			\textit{At this point both $A$ and $D$  dominate both $B$ and $C$.}\;
			\textbf{Phase Three} \;	
			\Numberline Run \textsc{Cut and Choose}  on the current residue for agents $B$ and $C$. \;}
		\caption{\textsc{Main Protocol}$(N)$}\label{alg:MainAlgorithm} 
	\end{algorithm}
} \bigskip 

Using Definition \ref{def:gain}, we identify a suballocation among $\allocation^1, \allocation^2, \allocation^3, \allocation^4$, where the gain of each agent is small compared to her combined gain from the other three suballocations (line \ref{line:main:find_allocation} of the algorithm). The existence of such an allocation is shown in Lemma~\ref{lem:pigeon_lemma}. 
Then, the redistribution is performed via the \correct protocol which takes as input an allocation $\allocation$, produced by $\core$, 
and outputs an allocation $\allocation' = \pi(\allocation)$, where $\pi$ is a permutation on $N$. In doing so, special attention is paid to the insignificant piece of $\allocation$. 
For now, we treat $\correct$ as a black box and  ask that it satisfies the three properties below; see Section~\ref{sec:correct} for its description and analysis.

\begin{restatable}{ctprop}{correctpropone}\label{ctproperty:insignificant}
	The insignificant piece of $\allocation$ is given to a different agent in $\allocation'$. In particular, it is given to an agent that has marked it in $\allocation$.
\end{restatable}

\vspace{-5pt}
\begin{restatable}{ctprop}{correctproptwo}\label{ctproperty:high value}
	If a non-cutter was  allocated her favorite unmarked piece in $\allocation$, she will again be allocated a piece of the same value in $\allocation'$. 
\end{restatable}

Assume there is no agent dominating everyone else, meaning that $\gain{\allocation}{i}$ is  defined for all $i\in N$. For a partial envy-free allocation like $\allocation$, the gain of any agent is nonnegative. However, this may not be true for $\allocation'$, as it is not necessarily envy-free. What we need is for $(\mathcal X \mysetminus \mathcal A)\cup \mathcal A'$ to be envy-free, and towards this $\gain{\allocation'}{i}$ should not be too small for any  $i\in N$.

\begin{restatable}{ctprop}{correctpropthree}\label{ctproperty:gain bound} 
	$\gain{\allocation'}{i} \geq - \gain{\allocation}{i}$ for all agents $i$. 
\end{restatable}

By Correction property \ref{ctproperty:insignificant}, the insignificant piece has changed hands after line \ref{line:main:correct1}. This allows us to make one extra call to \core in order to enforce one more domination (line \ref{line:main:core_again}).
Hence, the intermediate step is completed and we know that agent $1$ dominates at least 2 other agents. 
If she dominates all three of them, we can run any of the known procedures for 3 agents on the residue, and be done with only a few queries. The interesting remaining case is to assume that agent $1$ currently dominates exactly two other agents.

At this point there are various ways to proceed,  each with a different query complexity. E.g., we could repeat the whole process so far, but with agents $2$ and $3$ as cutters, and get at least 6 edges in the domination graph. This would ensure a node with in-degree two, but it requires several calls to \core. Instead, and quite remarkably, we show that it suffices to run \core only one more time, with the agent who is not dominated by agent $1$ as the cutter. As we prove in Lemma~\ref{lem:dom_by_D}, this makes the cutter dominate one of the agents that are dominated by agent $1$. Hence, phase one is now complete, as we have one agent with in-degree two.





\begin{remark}
	The intermediate step of getting a node with out-degree two has also been utilized in \cite{AM16a}. The goal there however was to make every agent dominate two other agents, whereas we only needed this to hold  for one agent.	
\end{remark}

\paragraph*{Phase Two.}
Suppose   phase two starts with a partial envy-free allocation where some agent, say $A$, is dominated by agents $B$ and $C$ (Figure~\ref{fig:desired_graph}). Our next goal is to produce a partial envy-free allocation where both $A$ and $D$ dominate both $B$ and $C$ (Figure~\ref{fig:phase_two}).
To achieve this goal, we execute $\core$ twice with $D$ as the cutter, i.e., with the agent not involved in the dominations of Phase one. Again, we need to argue about the behavior of \core under the existing dominations, and we ask for the following property.

\begin{restatable}{cprop}{corepropfour}\label{cproperty:domination}
	Assume we run \core with $D$ as the cutter, and suppose agent $A$ is dominated by the other two non-cutters, $B$ and $C$, neither of whom dominates the other. 
	Then, \emph{(1)} $A$ gets her favorite of the four complete pieces without making any marks, \emph{(2)} at least three complete pieces are allocated, and \emph{(3)} if a non-cutter, say $B$, gets a partial piece, then the remaining non-cutter, $C$, is indifferent between her piece and $B$'s piece.
\end{restatable}

Using this property, we can show that after one call to $\core$ (1st execution of line \ref{line:main:core_phase2}), agents $A$ and $D$ will both dominate either $B$ {\em or} $C$.  
However, we need domination over both $B$ {\em and} $C$. The second call to \core (2nd execution of line \ref{line:main:core_phase2}) ensures that we can again resort to  the $\correct$ protocol. 
If, after the two calls to $\core$, only one of $B$ and $C$, say $B$, is dominated by both $A$ and $D$, then running $\correct$ on one of the two core allocations from this phase---the one where the gain of $B$ is smaller---resolves the issue, and makes $A$ and $D$ dominate both $B$ and $C$.

\paragraph*{Phase Three.}
Since both agents $B$ and $C$ are dominated by $A$ and $D$, we just execute the  cut-and-choose protocol for $B$ and $C$, where $B$ cuts the  residue in two equal pieces and $C$ chooses her favorite piece. This completes our algorithm.

\subsection{Proof of Correctness of the Main Protocol}\label{sec:correctness}
In this section we analyze the correctness and the complexity of the main protocol. 
%
%
%
For our proof, we assume that  $\core$ and $\correct$ satisfy the properties mentioned earlier. 
For now, we take as granted the following theorems, which are proved in Sections \ref{sec:core} and \ref{sec:correct} respectively.
\begin{restatable}{theorem}{corethm}
	\label{thm:core theorem}
	The \core protocol in Section~\ref{sec:core} satisfies $\core$ properties~\ref{cproperty:whole-pieces},~\ref{cproperty:envy-free} and~\ref{cproperty:domination}, and makes at most $9$ cut queries and $15$ evaluation queries.
\end{restatable}

\vspace{-5pt}
\begin{restatable}{theorem}{correctthm}
	\label{thm:correct theorem}
	The \correct protocol in Section~\ref{sec:correct}  satisfies $\correct$ properties~\ref{ctproperty:insignificant}, \ref{ctproperty:high value} and~\ref{ctproperty:gain bound}, and makes no queries.
\end{restatable}

\noindent We start with one of the most important observations for our analysis: running $\core$ twice with the same  cutter creates an edge in the domination graph. Lemma~\ref{lem:core_domination}, as well as Lemma~\ref{lem:pigeon_lemma} below, has its counterpart in \cite{AM16a} concerning their core protocol. 

\begin{lemma}\label{lem:core_domination}
	Starting with an envy-free partial allocation, if $\core$ is executed two (not necessarily consecutive) times with the same agent, say $A$, as the cutter, 
	then in the resulting allocation $A$ dominates the agent who got the insignificant piece in the first execution.
	Moreover, if the insignificant piece was the only partial piece in the first execution, then the above domination is directly established at the end of the first execution.
\end{lemma}

\begin{proof}
	If only one piece $p$ is partially allocated in the first execution, let $p=p' \cup p_r$, where $p'$, $p_r$ are the allocated and unallocated parts of $p$ respectively. 
	It holds that $v_A(p')\leq 1/4 \cdot v_A(I)$, and, for all other pieces $q$ allocated in that execution, $v_A(q)=1/4 \cdot v_A(I)$.
	Therefore, $p'$ is the insignificant piece and is allocated to $B$. Moreover, for the total residue after the execution, $\residue$, we have that $\residue=p_r$.
	Combining the above, we get that $v_A(p' \cup \residue)=v_A(p' \cup p_r)=v_A(p)=1/4 \cdot v_A(I)$, and since the cutter was allocated value equal to $1/4 \cdot v_A(I)$, 
	she won't be envious of $B$ even if the latter where allocated the entire residue $\residue$. That is, $A$ dominates $B$.	
	
	If more than one piece was partially allocated, then by $\core$ property~\ref{cproperty:whole-pieces} exactly two pieces were partially allocated. 
	Therefore, the residue after the first execution is $\residue = r_1 \cup r_2$, where $r_1$, $r_2$ are the unallocated parts of the insignificant and the other partially allocated piece respectively.
	By the definition of $r_1, r_2$ we have $v_A(r_1)\geq v_A(r_2)$ and thus $v_A(R) \leq 2\cdot v_A(r_1)$. 
	Let $\residue'$ be the new residue after any possible in-between executions of \core. 
	Clearly, we have $v_A(\residue') \leq v_A(\residue) \leq 2\cdot v_A(r_1)$. 
	Finally, let $\residue''$ be the residue after the second execution of $\core$ with $A$ as cutter. 
	Then $v_A(\residue'') \leq \frac{2}{4}\cdot v_A(\residue')$, since two out of the four equal pieces are allocated whole by $\core$  property~\ref{cproperty:whole-pieces}.
	Combining the two inequalities we get $v_A(\residue'') \leq v_A(r_1)$.
	However, in the first execution $A$'s 
	piece was worth to her at least $v_A(r_1)$ more than $B$'s piece, and all the intermediate partial allocations between the two executions  are envy-free by $\core$ property~\ref{cproperty:envy-free}. This means that $A$ would not be envious even if the whole $\residue''$ was given to $B$.
\end{proof}

Next, we need the existence of a suitable input for $\correct$ at line \ref{line:main:find_allocation} of the main protocol. 
Recall that $\allocation^k = \{ p^k_1, p^k_2, p^k_3, p^k_4 \}$ is the allocation output by the $k$th execution of $\core$.
The following lemma is rather straightforward using a pigeonhole principle argument. 

\begin{lemma}\label{lem:pigeon_lemma}
	Suppose $\core$ is run $4$ consecutive times with agent $1$ as the cutter. 
	Then, there exists an allocation $\allocation^j \in \{ \allocation^1, \allocation^2, \allocation^3, \allocation^{4}\}$ such that for all agents $i\in N \mysetminus \{1\}$:
	$ \gain{\allocation^j}{i} \leq \sum_{\ell\in \{1, 2, 3, 4\}\mysetminus\{j\}} \gain{\allocation^{\ell}}{i}.$
\end{lemma}

\begin{proof}
	It is possible that $\gain{\allocation^j}{i} > \sum_{\ell \neq j} \gain{\allocation^{\ell}}{i}$ only if $\gain{\allocation^j}{i}$ is the maximum of $\gain{\allocation^1}{i}$, $ \dots, \gain{\allocation^{4}}{i}$.
	Let $M_i = \argmax_{j\in \{1,2,3,4\}} \gain{\allocation^j}{i}$, for $i \in \{2, 3, 4\}$. It suffices to show that for some $j \in \{1, \dots, 4\}$, $M_i \neq j$ for all $i \in \{2, 3, 4\}$. This, however, is straightforward since there are at least $4$ options for $M_i$ and $i$ only takes $3$ values. 
\end{proof}

%

Finally, the next lemma guarantees that lines \ref{line:main:done_w_double_dom}-\ref{line:main:finish_phase1} do create a second domination over an agent already dominated by agent $1$, if such a domination is not already there, without destroying envy-freeness of the overall allocation.

\begin{lemma}\label{lem:dom_by_D}
	Suppose we have an envy free allocation $\tallocation$ where an agent, $A$, dominates  two other agents, $B$ and $C$. If  we run \core with  the remaining agent $D$ as the cutter, excluding $A$ from competition, then $D$ will dominate  $B$ or $C$ in the resulting envy-free allocation.
\end{lemma}

\begin{proof}
	The exclusion of $A$ from competition greatly simplifies the execution of \core. It is easy to check that in this case \core is equivalent to the following algorithm:\vspace{-2pt}
	\begin{itemize}
		\item $D$ cuts the residue into 4 equal-valued pieces.
		\item If agents $B$ and $C$ have different favorite pieces, the agents choose a piece in the order $B, C, A, D$ and the algorithm terminates with a complete envy-free allocation. 
		\item Otherwise, $B$ and $C$ make a 2-mark on their common piece. Suppose $C$ has the rightmost mark on this piece (the other case is symmetric).
		\item $C$ gets the marked piece up to $B$'s mark.
		\item $B$ gets her second favorite piece.
		\item $A$ gets her favorite piece among the two remaining.
		\item $D$ gets the remaining piece.
	\end{itemize}
	Clearly $D$ now dominates $C$ since the initial allocation was  envy-free, $D$ is the cutter and $C$ got the unique partial piece. 
	It remains to show that the resulting allocation is still envy-free. First, $D$ as the cutter  gets a piece of equal value to her favorite piece.
	Moreover, it is straightforward  that $B$ and $C$ cannot envy each other, neither do they envy $A$ as they choose their pieces before her. Finally, $A$ was dominating both $B$ and $C$, hence she will continue to dominate them, and she chose her piece before $D$. 
\end{proof}

Given all the above ingredients, we are ready to prove our main theorem.

\begin{proofof}{Proof of Theorem~\ref{thm:main_protocol}}
	We first argue about the correctness of the \textsc{Main Protocol}. It suffices to prove that phases one and two terminate with the desired domination structures. 
	Then it is straightforward that the third phase results in an envy-free allocation. 
	\medskip
	
	\underline{\textit{Phase one:}} Lemma~\ref{lem:core_domination} guarantees that once the \textit{for} loop is completed in phase one, we have a domination graph where node $1$ has out-degree at least one. 
	
	We first  consider the case where at least 2 agents got an insignificant piece in lines \ref{line:for_loop}-\ref{line:for_loop_agent1}. Then Lemma~\ref{lem:core_domination} guarantees that after line \ref{line:main:core_again} agent $1$ dominates all those agents. If lines \ref{line:main:done_w_double_dom}-\ref{line:main:finish_phase1}
	are executed, Lemma~\ref{lem:dom_by_D} guarantees that some agent becomes dominated by both agents $1$ and $E$ and phase one is successfully completed. Otherwise, lines \ref{line:main:else:S-C}-\ref{line:main:S-C} are executed and the algorithm terminates returning a complete envy-free allocation. The latter holds because when agent $1$ dominates all other agents, we only need to divide the residue among them.
	What is left to be shown is that after line \ref{line:main:core_again} the overall (possibly partial) allocation is envy-free. Then, either the Selfridge-Conway protocol completes the allocation without introducing any envy, or Lemma~\ref{lem:dom_by_D}  guarantees that the overall allocation at the end of phase one is envy-free. 
	$\core$ property~\ref{cproperty:envy-free}, however, takes care of the envy-freeness after line \ref{line:main:core_again} because so far we have only executed \core with no excluded agents from competition (5 times). 

	The  remaining  case is when the same agent got the insignificant piece in all 4 iterations  in lines \ref{line:for_loop}-\ref{line:for_loop_agent1}, i.e., when the condition in  line \ref{line:main:agent1_dom2} 
	is true. 
	Among the 4 \core allocations produced in lines \ref{line:for_loop}-\ref{line:for_loop_agent1}, let $\allocation^j$ be an allocation satisfying the conditions of Lemma~\ref{lem:pigeon_lemma}. 
	The execution of $\correct$ with input $\allocation^j$ gives an allocation $\pi(\allocation^j)$ 
	where the insignificant piece with respect to $\allocation^j$ is given to a different non-cutter, due to $\correct$ property~\ref{ctproperty:insignificant}. 
	%
	%
	Since the insignificant piece has gone now to a different agent, Lemma~\ref{lem:core_domination} implies that running $\core$ one more time with agent $1$ as the cutter results in agent $1$ dominating a \emph{new} agent \emph{in addition to the old one}. 
	Thus, with respect to the target domination graph, we argue about lines \ref{line:main:done_w_double_dom}-\ref{line:main:S-C} like above. That is, either we get the desired domination pattern in lines  \ref{line:main:done_w_double_dom}-\ref{line:main:finish_phase1} or the algorithm  terminates  with lines \ref{line:main:else:S-C}-\ref{line:main:S-C}. We still have to show that the allocation right after line \ref{line:main:correct1} is envy-free. However,
	by the choice of $\allocation^j$ and $\correct$ property~\ref{ctproperty:gain bound}, for $\tallocation =\left( \cup_{\ell\in \{1, 2, 3, 4\}\mysetminus\{j\}} \allocation^{\ell}\right) \cup \pi(\allocation^j)$  we have
	$\gain{\tallocation}{i}\ge 0$  for all agents $i$, i.e., $\tallocation$ is envy-free. $\core$ property~\ref{cproperty:envy-free} and Lemma~\ref{lem:dom_by_D} guarantee that the allocation at the end of phase one is envy-free as well. \medskip
	
	%
	%
	%
	%
	
	\underline{\textit{Phase two:}} Having obtained a partial allocation where one agent, say $A$, is dominated by two others, say $B$ and $C$, we execute $\core$ twice on the current residue with the remaining agent $D$ as the cutter (line \ref{line:main:core_phase2} of the main protocol). First, suppose that there is at least one agent excluded from competition. Without loss of generality, we may assume $C$ is such an agent. Then it is very easy to check (see \core in the next section) that this execution of core is equivalent to having $D$ cut in four equal pieces, and agents pick in order $A$, $B$, $C$, $D$, thus resulting in a complete allocation of the residue. Taking into account the dominations of $B$ over $A$ as well as of $C$ over $B$ and $A$, we see that such an allocation is also envy-free.
	
	We conclude that, if line \ref{line:main:B xor C} is reached, there was no domination between $B$ and $C$ before each of the last two $\core$ executions. That is, these executions excluded no one from competition.  
	Therefore, we can combine $\core$ property~\ref{cproperty:domination} (which guarantees there is only one partial piece) with the fact that the  allocation so far is envy-free to
	guarantee that the first $\core$ execution resulted in $A$ and $D$ both
	dominating one of $B$ and $C$, specifically, the one who got the only partial piece (if there is no such piece, the algorithm terminates with a complete envy-free allocation). To see this, suppose $B$ is the agent who received the partial piece. By $\core$ property~\ref{cproperty:domination}, we know that $A$ gets her favorite piece, $D$ gets a complete piece (equivalent to her favorite piece), and the left-over residue is derived only from the partial piece of $B$. Hence, $A$ and $D$ cannot be envious even if the whole residue is given to $B$.

	A possible complication now is that the second $\core$ execution might result again in  $B$ getting the partial piece. 
	Let $\tallocation$ be the overall partial allocation so far, i.e., at line \ref{line:main:B xor C}.
	We execute $\correct$ on the \core allocation $\allocation$ (out of the two last ones) where $\gain{\allocation}{B}$ is smaller. 
	By $\core$ property~\ref{cproperty:domination}, we have that agent $A$ never marked a piece in $\allocation$. But neither did $D$, since she is the cutter.
	Thus, by $\correct$ property~\ref{ctproperty:insignificant}, the unique partial piece of $\allocation$ goes to  $C$ in $\allocation' = \pi(\allocation)$.
	At the same time $A$ and $D$ received the value of their favorite complete piece in both corresponding executions of \core, following $\correct$ properties \ref{ctproperty:high value} and \ref{ctproperty:gain bound} (note that $\gain{\allocation}{D}=0$, as $D$ is the cutter).
	Combining this with the fact that $B$ and $C$  each received the unique partial piece of one suballocation, we get that $A$ and $D$ now dominate both $B$ and $C$.
	
	Finally, we argue about the overall partial allocation being envy-free. Clearly, $A$ and $D$ are non-envious. By 
	the choice of $\allocation$, we have $\gain{\tallocation'}{B} \geq 0$, where $\tallocation' = (\tallocation \mysetminus \allocation) \cup \pi(\allocation)$,
	thus $B$ is not envious. On the other hand, $C$ is not envious since by $\core$ property~\ref{cproperty:domination} she was indifferent between her piece and $B$'s piece in the first place.
	\medskip
	%
	
	\underline{\textit{Counting Queries:}} In the worst case, at most $8$ calls to \core are required followed by a call to \textsc{Cut and Choose}. Theorem \ref{thm:core theorem} directly gives an upper bound $8\cdot 9+1 =73$ cut  and $8\cdot15+1=121$ evaluation queries. For a more detailed argument matching the statement of the theorem, see the last part of the next section.
	\qed
\end{proofof}

\section{The $\core$ Protocol}\label{sec:core}


An important building block of  the whole algorithm is the \core protocol, used for allocating part of the current residue every time it is called. We begin with a high-level idea of how \core works.  It takes as input an agent, specified as the cutter, the current residue, and the current partial allocation. 
$\core$ first asks the cutter to divide the residue into four equally valued contiguous pieces. 
%
The cutter is going to be the last one to receive one of these four pieces. 
Regarding the remaining three agents, 
each of them will either be immediately allocated her favorite piece or will be asked to place a mark on certain pieces, 
based on the relative rankings of the non-cutters for the pieces, and on possible domination relations that have already been established. 
Marks essentially provide limits on how to partially allocate pieces that are desired by many agents, so that they can be given without introducing envy.

As seen in the pseudocode description of \core, there are two possible types of marks that can be placed; $2$-marks and $3$-marks. 
The type of mark that the agents will be asked to place depends mainly on the conflicts that arise for the favorite and second favorite pieces of each agent. 
The conditions that determine whether an agent will be asked to place a $2$-mark or a $3$-mark are described in lines \ref{line:core:mark condition1} and \ref{line:core:mark condition2} of the core protocol. 
These conditions simplify significantly the numerous cases that arise in the core protocol of \cite{AM16a}.
To describe the protocol, we need to formalize conflicts between agents for certain pieces. 

\bigskip
\begin{algorithm}[H]
	\DontPrintSemicolon 
	\NoCaptionOfAlgo
	{\small
		Agent $k$ cuts the current residue $R$ in four equal-valued pieces (according to her).\label{line:core:cut4}\;
		Let $S = N\mysetminus(\{k\}\cup \mathcal E)$ be the set of agents who may compete for pieces.\label{line:core:competition}\;
		\If{there exists $j \in S$ who has no competition in $S$ for her favorite piece} 
		{$j$ is allocated her favorite piece and is removed from $S$. 
		}\label{line:core:favorite}
		\If{every agent in $S$ has a different favorite piece\label{line:core:exploit-order-start}}{Everyone gets her favorite piece and the algorithm terminates.\label{line:core:exploit-order-finish}} 
		\For{every agent $i \in S$ \label{line:core:mark loop}}{
			\If{(1) $i$ has no competition for her second favorite piece $p$, \emph{\textbf{or}} \label{line:core:mark condition1} \;
				(2) $i$ has exactly one competitor $j \in S$ for $p$, $j$ also considers $p$ as her second favorite, and $i$,$j$ each have exactly one competitor for their favorite piece \label{line:core:mark condition2}}
			{$i$ makes a $2$-mark.}
			\Else{$i$ makes a $3$-mark.} \label{line:core:exploit-order-end}
		}
		Allocate the pieces according to a rightmost rule:\;
		\If{an agent has the rightmost mark in two pieces\label{line:core:RM-in-two} } 
		{Out of the two partial pieces, considered until the second rightmost mark (which always exists by Lemma~\ref{lem:marked} below), she is allocated the one she prefers. \label{line:core:choose}\;
			The other partial piece is given to the agent who made the second rightmost mark on it.\label{line:core:RM1}}
		\Else{\label{line:core:else} Each partial piece is allocated---until the second rightmost mark---to the agent who made the rightmost mark on that piece.}\label{line:core:RM2}
		\If{any non-cutters were not given a piece yet}
		{Giving priority to any remaining agents in $S$ (but in an otherwise arbitrary order), they choose their favorite unallocated complete piece.}\label{line:core:RM3}
		The cutter is given the remaining unallocated complete piece. \label{line:core:cutter}
		\caption{$\core\,(k, R, \tallocation, \mathcal E)$} \label{fig:CoreAlgorithm} 
	}
\end{algorithm} \medskip  

\begin{definition}
	\label{def:competition}
	During an execution of \core, let $P$ be a set of pieces and $S$ be a subset of non-cutters. We say that an agent $i\in S$ has {\em competition} for a piece $p\in P$, if \emph{(1)} $i$ is not dominated by everyone in $S$, and \emph{(2)} there exists $j\in S$ such that $p$ is  $j$'s favorite or second favorite piece in $P$. We call $j$ a {\em competitor} of $i$ for the piece  $p$.
\end{definition}

Definition \ref{def:competition} helps us identify whether we need to perform a $2$-mark or $3$-mark on the available pieces.
Furthermore, in some cases where we know that certain domination patterns appear, it is convenient to prevent some agents from competing for any piece. 
Hence, \core also takes as an  input a subset $\mathcal E$ of  agents that are  \emph{excluded from competition} in line \ref{line:core:competition}. 
In most cases, this argument is the emptyset, with the exception of lines \ref{line:main:finish_phase1} and \ref{line:main:core_phase2} of the \textsc{Main Protocol}.


The main result for \core, which is crucial for the entire algorithm to work, is the next theorem. \medskip

\begin{rtheorem}{Theorem}{\ref{thm:core theorem}}
	The \core protocol satisfies $\core$ properties~\ref{cproperty:whole-pieces},~\ref{cproperty:envy-free} and~\ref{cproperty:domination}, and makes at most $9$ cut queries and $15$ evaluation queries.
\end{rtheorem}\medskip


%
%



The proof of Theorem \ref{thm:core theorem}  is based on a series of lemmas regarding the properties of \core.
We start by establishing the following key lemma. 


\begin{lemma}\label{lem:marked}
	Let $\{p_1, p_2, p_3, p_4 \}$, be the four pieces created by the cutter in the initial step of \core. After all markings have taken place,
	at most two pieces have marks and each marked piece has at least two marks.
\end{lemma}

\begin{proof} 
	We prove the statement for the case where no agent is excluded from competition. Note that when someone is excluded from competition, the lemma is either straightforward (at least two agents excluded) or reduces to case (1.) below.
	
	To facilitate the proof we use a convenient matrix notation to describe instances. E.g., \vspace{-3pt}
	{\footnotesize 	\begin{center}
			\setlength{\tabcolsep}{0.5em} 
			{\renewcommand{\arraystretch}{1.2}
				\begin{tabular}{|c|c|c|c|}
					\hline 1 & 2 &   &  \\ 
					\hline 1 & 3 & 2 & 4 \\  
					\hline 2 &   & 1 &  \\ 
					\hline
			\end{tabular}}
	\end{center}}
	\noindent Each row corresponds to a non-cutter, call them $A, B, C$. Each column corresponds to a piece: $p_1$, $p_2$, $p_3$ and $p_4$. The number in cell $(i,j)$ indicates the rank of piece $j$ for agent $i$. A blank cell $(i,j)$ means ``any of the remaining options''. We say that two instances are \textit{isomorphic} if their matrix notation is the same, up to renaming of the agents and pieces. We consider four cases:
	\begin{enumerate}[leftmargin=*]
		\item Line~\ref{line:core:favorite} is executed and some agent $j$ is removed. If the remaining agents have a different favorite piece, then the lemma is vacuously true. If the remaining agents have the same favorite piece, then one of conditions~\ref{line:core:mark condition1} or~\ref{line:core:mark condition2} will hold, and thus they each make a $2$-mark on their favorite piece. 
		\medskip
		
		\item Line \ref{line:core:favorite} is not executed and all agents have a different favorite piece; the lemma is clearly true since there are no marked pieces in this case.
		\medskip
		
		\item Line \ref{line:core:favorite} is not executed and exactly two agents have the same favorite piece. The instance is isomorphic to: \vspace{-3pt}
		{\footnotesize 		\begin{center}
				\setlength{\tabcolsep}{0.5em} 
				{\renewcommand{\arraystretch}{1.2}
					\begin{tabular}{|c|c|c|c|}
						\hline 1 &   &   &  \\ 
						\hline 1 &   & {\ } & {\ } \\  
						\hline   & 1 &   &  \\
						\hline 
				\end{tabular}}
		\end{center}}
		
		\begin{itemize}
			\item If agent $C$ has no competition for her favorite piece, she takes it and leaves. The remaining agents make exactly one 2-mark each, on the first piece.
			\smallskip
			
			\item If agent $C$ has competition with exactly one agent for her favorite piece: 
			\vspace{4pt}
			{\footnotesize 			\begin{center}
					\setlength{\tabcolsep}{0.5em} 
					{\renewcommand{\arraystretch}{1.2}
						\begin{tabular}{|c|c|c|c|}
							\hline 1 & 2  &   &  \\ 
							\hline 1 &   & 2  &  {\ }  \\  
							\hline   & 1 &   &  \\
							\hline 
					\end{tabular}}
			\end{center}} 
			\vspace{4pt}
			
			In this case, $A$ makes a 3-mark, marking $p_1$ and $p_2$, and $B$ makes a $2$-mark, marking $p_1$. Therefore, $p_1$ will definitely have at least $2$ marks. $C$ can make a $2$-mark or a $3$-mark, depending on which is her second favorite piece, but she will definitelly make a mark on $p_2$ (and thus $p_2$ will have exactly two marks on it). It remains to show that no other piece is marked. If $C$'s second favorite piece is $p_3$ or $p_4$, she makes a $2$-mark. Otherwise, she makes a $3$-mark, but her second mark is on $p_1$.
			\smallskip
			\item If agent $C$ has competition with exactly two agents for her favorite piece: 
			\vspace{4pt}
			{\footnotesize 			\begin{center}
					\setlength{\tabcolsep}{0.5em} 
					{\renewcommand{\arraystretch}{1.2}
						\begin{tabular}{|c|c|c|c|}
							\hline 1 & 2  &   &  \\ 
							\hline 1 & 2  &  {\ }   &  {\ }  \\  
							\hline   & 1 &   &  \\
							\hline 
					\end{tabular} }
			\end{center}} \medskip 
			
			Agents $A$ and $B$ make $3$-marks, on $p_1$ and $p_2$, and thus at least two pieces are marked with $2$ marks each. It remains to show that no other piece is marked. If $C$'s second favorite piece is $p_3$ or $p_4$ she makes a $2$-mark on $p_2$. Otherwise, she makes a $3$-mark on $p_1$ and $p_2$.			
		\end{itemize}
		\item Line \ref{line:core:favorite} is not executed and all agents have the same favorite piece. 
		If they also have the same second favorite piece, then they all make $3$-marks on the same two pieces. 
		If they have different second favorite pieces, they all make $2$-marks on $p_1$. Otherwise, the instance is isomorphic to:  %
		\vspace{2pt}
		{\footnotesize 		\begin{center}
				\setlength{\tabcolsep}{0.5em} 
				{\renewcommand{\arraystretch}{1.2}
					\begin{tabular}{|c|c|c|c|}
						\hline 1 & 2  &   &  \\ 
						\hline 1 & 2  &   &  {\ }  \\  
						\hline 1 &    & 2 &  \\
						\hline 
				\end{tabular}}
		\end{center}}\vspace{5pt}
		In this case, $A$ and $B$ make a $3$-mark on $p_1$ and $p_2$, while $C$ makes a $2$-mark on $p_1$: two pieces marked with at least $2$ marks each. \qedhere
	\end{enumerate}
\end{proof}

An almost immediate corollary is the following:
\begin{corollary}\label{cor:unmarked pieces}
	All pieces allocated in lines~\ref{line:core:RM3} and \ref{line:core:cutter} of the \core protocol are unmarked, and therefore they are allocated as complete pieces.
\end{corollary}

\begin{proof}
	If there existed two marked pieces and the same agent had the rightmost mark in both, then they will both be allocated in line~\ref{line:core:RM1}. By Lemma~\ref{lem:marked}, there are no other marked pieces, hence the pieces that have remained in lines~\ref{line:core:RM3} and \ref{line:core:cutter} are unmarked.	
	Otherwise, the else part in line \ref{line:core:else} is executed (no agent has the rightmost mark in two pieces), and all partial pieces are allocated in line~\ref{line:core:RM2} of the algorithm. Hence, again, any pieces that have remained when the algorithm goes beyond line~\ref{line:core:RM2} are unmarked.	
\end{proof}

Given Lemma~\ref{lem:marked}, and since only marked pieces are allocated partially, 
it follows that the cutter and at least one other agent receive complete pieces, each of which the cutter values as $1/4$ of the input residue, thus establishing \core property $\ref{cproperty:whole-pieces}$:
\corepropone* 
Towards proving the remaining \core properties, we first establish the following two lemmata. Note that the lemmata still hold when some agents are excluded from competition.

\begin{lemma}\label{lem:x unmarked}
	If all non-cutters have placed their marks, as dictated by the \core protocol, and some agent $i$ has marked a piece $p$, 
	her value for $p$ up to her mark is equal to her value for her favorite unmarked piece.
\end{lemma}

\begin{proof}
	Assume  agent $i$ has made an $x$-mark. We  need to prove that  her $x$-th favorite piece has remained unmarked after all agents have placed their marks.
	Suppose this is not the case. If $i$ had made a $2$-mark and someone marked her second favorite piece, it means she had competition for it. 
	Thus, the condition in line \ref{line:core:mark condition2} must have been true. But the condition would be true for her single competitor $j$ as well. 
	So $j$ would also make a $2$-mark, leaving their common second favorite piece unmarked, leading to a contradiction. 
	On the other hand, if $i$ had made a $3$-mark and some agent $j$ marked her third favorite piece, 
	then at least three pieces would have been marked: $i$'s most, second, and third favorite piece. By Lemma~\ref{lem:marked}, this is again a contradiction.
\end{proof}

\begin{lemma}\label{lem:I get what I mark}
	If agent $i$ has made an $x$-mark, $x \in \{2,3\}$, she receives a piece with value at least equal to the value she has for her $x$-th favorite piece 
	out of those that were still unallocated when she made her marks.
\end{lemma}

\begin{proof}
	First, we prove the statement for the case where $i$ gets a partial piece, i.e., $i$ gets her piece in line~\ref{line:core:RM1} or line~\ref{line:core:RM2}. If $i$ makes a $2$-mark on a piece $p$ and she is allocated a part of $p$ in either step, she either gets the part of $p$  up to her mark (when she is the ``other agent'' of line~\ref{line:core:RM1}) or a part of $p$ beyond her mark, i.e., of value at least that of her second favorite (original) piece by the definition of a $2$-mark. The argument for agents with 3-marks who get partial pieces is identical. 
	
	Next, we prove it for the case where $i$ is allocated a complete piece in line~\ref{line:core:RM3}, i.e., she did not have the rightmost mark on her favorite piece, 
	nor was she the ``other agent'' of line~\ref{line:core:RM1}. Then, by Lemma \ref{lem:x unmarked}, $i$'s $x$-th favorite piece has remained unmarked and will be allocated completely.
	It remains to show that $i$ will be the one who gets that piece. If there is no other agent, apart from $i$ and the cutter, remaining in line~\ref{line:core:RM3}, $i$ will surely get her $x$-th favorite piece. 
	Suppose there is another agent $j$ remaining in line~\ref{line:core:RM3}. 
	Since at this point only unmarked pieces are unallocated, and three agents---$i$, $j$, and the cutter---are yet to receive a piece, 
	it follows that only a single piece was marked, i.e., both $i$ and $j$ made $2$-marks on the same piece. 
	But since neither $i$ nor $j$ got their favorite piece, there must have been a third agent competing for it. 
	Therefore, condition (2) in line \ref{line:core:mark condition2} does not apply, and the only way $i$ and $j$ are qualified for making a $2$-mark, is if they do not have competition for their second favorite piece (line~\ref{line:core:mark condition1}). 
	No competition for their second favorite piece, however, implies non-conflicting preferences in line~\ref{line:core:RM3}.
\end{proof}

To continue with the analysis, we need to understand what could possibly cause agents to experience envy during the execution of the protocol. 
For this,
suppose that an agent $i$ has made a mark on a piece. 
If an agent $j \neq i$ is allocated this piece strictly beyond $i$'s mark, Lemma \ref{lem:I get what I mark} is not enough to ensure that $i$ will remain envy-free.

\begin{lemma}\label{lem:no one takes my shit pt2}
	If agent $i$ marked a piece $p$ and she does not have the rightmost mark in two pieces, then no agent $j\neq i$ will be allocated $p$ strictly beyond $i$'s mark.
\end{lemma}
\begin{proof}
	Since $p$ has a mark on it, by Corollary~\ref{cor:unmarked pieces}, it could not have been allocated in line~\ref{line:core:RM3} or \ref{line:core:cutter}. 
	If $i$ has the rightmost mark on $p$, and since $i$ does not have the rightmost mark in another piece, it is her who gets (a part of) $p$ in line~\ref{line:core:RM2}. 
	If $i$ does not have the rightmost mark on $p$, then $p$ can be allocated to another agent in line \ref{line:core:choose}, or \ref{line:core:RM1}, or \ref{line:core:RM2}. But in these lines, pieces are allocated up to the second rightmost mark, 
	hence, no agent can get a part of $p$ strictly beyond $i$'s mark.
\end{proof}

We are now ready to prove \core properties $\ref{cproperty:envy-free}$ and $\ref{cproperty:domination}$. 

\coreproptwo* 

\begin{proof}
	Envy-freeness for an agent $i$ who either is the cutter, or was allocated her favorite piece completely is straightforward: 
	If $i$ is the cutter, she considers all pieces to be equal and is allocated one of them completely, while the other agents are each allocated at most one complete piece.
	If $i$ is not the cutter but she still got her favorite piece completely, she is envy-free for the same reason. 
	
	Let $i$ be a non-cutter who did not get her complete favorite piece. This means she was asked to make an $x$-mark, which, by definition, means that pieces left unmarked by $i$ have value at most her value for her $x$-th favorite piece. Let $v_x$ be this value.
	By Lemma \ref{lem:I get what I mark}, $i$ gets a piece of value $v \geq v_x$.
	The pieces which are allocated completely are a subset of the pieces not marked by $i$, thus they all have value at most $v_x$ to her.
	Therefore, $i$ will not be envious of any agent getting a complete piece. 
	Now, if $i$ did not have the rightmost mark in two pieces, no partial piece will be allocated beyond her mark by Lemma~\ref{lem:no one takes my shit pt2}.
	This means that such pieces will have a value of at most $v_x$ to her.
	Finally, if $i$ had the rightmost mark in two pieces, she will be given her favorite among the two, as seen in line~\ref{line:core:RM1}, so she will not envy the agent who gets the other. 
	This completes the proof of envy-freeness. 
\end{proof}

%

\corepropfour* 

\begin{proof}
	Agent $A$ is allocated her favorite piece completely in Line \ref{line:core:favorite}, since she is dominated by $B$ and $C$ and has no competition, thus satisfying (1). 
	Towards (2), in case the remaining non-cutters, $B$ and $C$, have no competition for their favorite piece out of those remaining, 
	each is allocated her favorite piece, the cutter is allocated the last piece and the algorithm terminates. 
	Otherwise, they each have at most one competitor for their second favorite piece, 
	therefore one of the conditions of lines \ref{line:core:mark condition1} or \ref{line:core:mark condition2} is met. 
	This directly implies that only one piece will be marked, their common favorite, and since only marked pieces are allocated partially, three pieces will be allocated completely.
	Finally, towards (3), the condition of line~\ref{line:core:RM-in-two} is not valid since there is only one marked piece, therefore line~\ref{line:core:RM2} is executed. 
	Agent $B$ must have  the rightmost mark on the marked piece and the only candidate for the second rightmost mark is agent $C$. 
	Since $B$ gets the piece up to the second rightmost mark, 
	i.e., $C$'s 2-mark, and $C$ gets her second favorite piece (excluding the piece given to $A$) completely, $C$ has the same value for her own and $B$'s allocated piece.
\end{proof}


\subsubsection*{Counting Queries}



Recall the discussion in Section \ref{sec:prelims} about the residue being a finite union of intervals and not a single interval. It is important that our algorithm should know, at any time during its execution, the values of all the agents for all the intervals that make up the residue. Otherwise, our queries cannot be simulated by simple queries of the Robertson-Webb model. E.g., a query may ask to evaluate the pieces $[0.2,0.3]$ and $[0.3,0.5]\cup[0.7,0.8]$ of the residue $[0.2,0.5]\cup[0.7,0.8]$. In fact, without any prior knowledge, one needs 3 simple evaluation queries for that;  one query for $[0.2,0.3]$,  one for $[0.3,0.5]$, and one for $[0.7,0.8]$. However, assume we already knew the values of $[0.2,0.5]$ and $[0.7,0.8]$. Then we only need one query for $[0.2,0.3]$.
During the execution of our protocols we want each of our ``queries'' to actually correspond to only one Robertson-Webb query. So we make sure that the relevant information is always available. In particular, when cutting or evaluating consecutive pieces, we do it from left to right, and by the end of each execution of \core we learn the value of any agent for any trim of a partially allocated piece.

When we ask agents to place a mark on a piece, this corresponds to 1 cut query, since we can simulate this action with a cut query. 
Thus $\core$ requires $15$ evaluation and $9$ cut queries. We assume that in the beginning of the current execution all agents know the values of all the intervals that make up the residue.

\begin{itemize}
	\item The cutter is asked to cut the residue into four pieces in line \ref{line:core:cut4}. ($3$ cut queries)
	\item The ordering each non-cutter has for the pieces is necessary for lines \ref{line:core:exploit-order-start}-\ref{line:core:exploit-order-end}, therefore all three non-cutters are queried for their value for each of the four pieces. ($3\cdot 3 = 9$ evaluation queries)
	\item According to the conditions in lines~\ref{line:core:mark condition1} and~\ref{line:core:mark condition2}, each of the non-cutters remaining (at most three) will make either a 2-mark or a 3-mark. Worst case, three agents remain and all make a 3-mark. 
	($3\cdot 2 = 6$ cut queries)
	\item We ask evaluation queries so that every agent learns the value of each marked piece up to the second rightmost mark. Worst case, there are two marked pieces and for each we asked everyone but the agent that placed the second rightmost mark. ($3\cdot 2 = 6$ evaluation queries)
\end{itemize}
Note that the protocol now has all the information needed to check the condition on line \ref{line:core:RM1}, find the values of the partial pieces, calculate gain, and know the values of all the intervals that make up the new residue.

It should be noted that in a few cases we know we have less queries. In particular, the calls to \core in lines \ref{line:main:finish_phase1} and \ref{line:main:core_phase2} of the \textsc{Main Protocol}, are guaranteed to produce at most one partial piece. Therefore, in  these 3 calls we only need $2$ cut queries and $3$ evaluation queries for everyone to learn the value of the marked piece up to the second rightmost mark. In total, each of these special cases of $\core$  requires $12$ evaluation and $5$ cut queries instead of $15$ evaluation and $9$ cut queries.

Moreover, on the second execution of line \ref{line:main:core_phase2}, agents $A$ and $D$ (as they are called in the \textsc{Main Protocol}) never need to know  the value of the marked piece up to the second rightmost mark.

In particular, the total number of queries of the \textsc{Main Protocol} is  $5\cdot 9+3\cdot 5+1 =61$ cut  and $5\cdot 15+3\cdot 12 +1 -2 =110$ evaluation queries.

\section{The $\correct$ Protocol}\label{sec:correct}

\correct takes as input an allocation $\allocation$, produced by a single execution of \core. It outputs a redistribution of the pieces $\allocation' = \pi \left( \allocation \right)$
such that the insignificant piece goes to a different agent. 
Some of the pieces are marked (the ones partially allocated  by $\core$),  while others are unmarked (the pieces  allocated completely).
Note that each of the marked pieces of $\allocation$ still has exactly two marks, the rightmost and the second rightmost marks of the original untrimmed piece of the execution of \core that produced $\allocation$. This second rightmost mark coincides with the left endpoint  of the allocated piece as it appears in allocation $\allocation$. \smallskip

The redistribution should satisfy further properties, so that both the envy-freeness of the overall partial allocation and certain dominations are preserved.
Towards bounding the ``local'' envy that the redistribution may cause, the notion of \textit{gain} (see Section $\ref{sec:prelims}$)  is crucial.

In the pseudocode description below, we refer to the cutter in allocation $\allocation$ as $D$, the non-cutter who holds the insignificant piece as $A$, 
the non-cutter who gets the insignificant piece after executing $\correct$ as $B$, and the remaining non-cutter as $C$. \bigskip

\begin{algorithm}[H]
	\DontPrintSemicolon
	\NoCaptionOfAlgo
	{\small
		Let $A$, $B$ be the agents having the two marks on the insignificant piece, and suppose $A$ was given this piece in allocation $\allocation$.\;
		The insignificant piece is allocated to $B$.\;\label{line:correction:significant}
		\If{there is no other partial piece\label{line:correction:case1}}{Agents choose their favorite piece in the order $C$, $A$, $D$.\label{line:correction:case1 content}}
		\Else{
			Find the rightmost mark not made by $B$ on the other partial piece. Let $E\in\{A, C\}$ be the agent who made it.\;
			Agent $E$ is allocated the partial piece.\;\label{line:correction:partial}
			The last non-cutter chooses her favorite among the two complete pieces.\;\label{line:correction:final}
			The cutter is allocated the remaining (complete) piece.\; \label{line:correction:last}
		}
		\caption{\textsc{$\correct(\allocation)$}} \label{fig:CorrectionAlgorithm} 
	}
\end{algorithm} \bigskip

The main result about \correct is the next theorem. The remaining of this section is dedicated to its proof.


\bigskip

\begin{rtheorem}{Theorem}{\ref{thm:correct theorem}}
	The \correct protocol satisfies $\correct$ properties~\ref{ctproperty:insignificant}, \ref{ctproperty:high value} and~\ref{ctproperty:gain bound}, and makes no queries.
\end{rtheorem}\medskip

We stick to the notation we used in the description of \correct. That is, given an allocation $\allocation$, $D$ is the cutter, $A$ is the non-cutter who holds the insignificant piece, $B$ is the non-cutter who gets the insignificant piece after executing $\correct$ on $\allocation$, and $C$ is the remaining non-cutter.

\correct property~\ref{ctproperty:insignificant}  trivially holds, since in line \ref{line:correction:significant} of $\correct$ the insignificant piece is allocated to agent $B$:
\correctpropone*

\correctproptwo*
%
%

\begin{proof}
	If $i$ is a non-cutter who received her favorite unmarked piece in $\core$, she has made no marks. 
	Therefore, since by definition $A$ and $B$ have made marks on the insignificant piece and $D$ is the cutter, 
	$C$ is the only non-cutter agent who could have formerly been allocated her favorite piece.
	We need to show that in this context $C$ will be the first to choose one of the complete pieces in $\correct$. 
	If we are in the case of line \ref{line:correction:case1}, this follows immediately. 
	Otherwise, since $C$ has made no marks, she is not allocated the other partial piece, 
	and is thus the ``final non-cutter'' of line \ref{line:correction:final} who chooses first among the two unmarked pieces. 
\end{proof}

\correctpropthree*

For the proof of \correct Property~\ref{ctproperty:gain bound} we assume that $\allocation$ is produced by a call to  \core where  no agent was excluded from competition. Note that this is always the case when \correct is used in our main protocol (see the proof of Theorem \ref{thm:main_protocol}).
We  need Lemma \ref{lem:secondfav} below, but first it is helpful to establish some tie-breaking conventions for what follows. Whenever a marked and a complete piece have the same value, we will consider the marked piece to be better from the agent's perspective. Moreover, if an agent has the second rightmost mark in two pieces, then we arbitrarily consider one to be her favorite, and the other her second favorite.


\begin{lemma}\label{lem:secondfav}
	Any non-cutter $i$ who was not allocated her favorite whole piece when $\allocation$ was produced, receives in $\correct(\allocation)$ value at least equal to that of her  favorite piece in $\allocation$ among those formerly allocated to agents in $N \mysetminus (\{i\}\cup D_i)$. 
\end{lemma}

\begin{proof} 
	Let $S = N \mysetminus (\{i\}\cup D_i)$. In most cases we will show the stronger---but easier to prove---property that $i$'s new allocation has value to her at least equal to that of her second favorite piece among all four pieces in $\allocation$.
	
	We start by proving this for the case when $\correct$ allocates a partial piece $p$ to  $i$. 
	Note that this may happen only if $i$ has a mark on $p$ (lines \ref{line:correction:significant}, \ref{line:correction:partial}).
	Thus $i$'s value for $p$ up to her mark is equal to her value for her favorite unmarked piece, and therefore there is only one piece which $i$ might value more than $p$: 
	the other marked piece (if it exists). That is, $i$ is allocated her overall favorite or second favorite piece.
	
	Next, we move to the case where $i$ is allocated an unmarked piece $p$ by $\correct$, distinguishing among the two sub-cases where that could happen: 
	(a) in line \ref{line:correction:case1 content}, and (b) in line \ref{line:correction:final}, where there are one and two marked pieces respectively.
	
	For sub-case (a), we claim that it suffices to show that $i$ is allocated value at least equal to that of her favorite unmarked piece
	out of those formerly allocated to agents in $S\cup \{i\}$. To see that, note that the fact that there is only one marked piece, 
	immediately implies that she is allocated value at least equal to that of her second favorite piece out of those formerly given to agents in $S\cup \{i\}$.
	Since the input contains only one marked piece, nobody has made a $3$-mark. 
	Moreover, since agents $A$ and $B$ have marks on the marked piece, they both made a $2$-mark. Clearly, $B$ receives as much value as possible, given $\allocation$, by getting the marked piece up to her mark.
	As for agent $C$, she either had made a $2$-mark as well, or she had been allocated her favorite whole piece.
	If the former occurred: $A$, $B$ and $C$ had the same favorite piece (part of which now constitutes the insignificant piece), 
	and all made $2$-marks. Since each clearly had more than one competitor for her favorite piece, for all three the $\core$ condition of line (\ref{line:core:mark condition1})
	must have been true, i.e., they all had different second favorite pieces. Therefore, in the input of $\correct$, each has a different favorite unmarked piece and each of $A$ and $C$ is allocated her favorite unmarked piece.
	If the latter occurred: $C$ had no competition from $A$. So, either $C$'s favorite piece is $A$'s third or fourth favorite, or $C$ is dominated by $A$. 
	Therefore, after $C$ chooses her favorite piece (of the three unmarked ones in line \ref{line:correction:case1 content}), 
	$A$'s favorite unmarked piece
	out of those formerly allocated to agents in $S\cup \{i\}$ is still unallocated and she can take it.
	
	For sub-case (b), first notice that
	if $i$ does not have a mark in both pieces, then her favorite unmarked piece (i.e., the one she  chooses in line \ref{line:correction:final}) is either her overall favorite or second favorite piece.
	On the other hand, if $i$ has a mark in both marked pieces, one of them is her favorite and the other is her second favorite.
	Since she weren't allocated the non-insignificant marked piece, she must have had the second rightmost mark on it, 
	which means she considers it equal to her favorite unmarked piece $p$ that she  chooses in line \ref{line:correction:final}. 
	Thus, if the insignificant marked piece is her favorite, then the other marked piece is her second favorite and so $p$ has value equal to her second favorite piece. 
	If not, then the non-insignificant marked piece is her favorite, and so $p$ has value equal to that of her favorite piece. 
	Either way she is allocated value at least equal to that of her second favorite piece. 
\end{proof}\smallskip

\begin{proofof}{Proof of \correct Property \ref{ctproperty:gain bound}}
	Towards proving $\correct$ property~\ref{ctproperty:gain bound}, recall Definition \ref{def:gain}. Given a partial allocation $\tallocation$ and a suballocation $\allocation$ of $\tallocation$, let $i$ be an agent who dominates agents in $D_i\subseteq N$. If $p_i$ is agent $i$'s piece in $\allocation$ and $s_i$ is $i$'s favorite piece out of those allocated to agents in $N \mysetminus (\{i\}\cup D_i)$, then $\gain{\allocation}{i} = v_i(p_i) - v_i(s_i)$.
	
	
	Now assume $\allocation$ was produced by $\core$ (and recall that no agent was
	excluded from competition). By $\correct$ property~\ref{ctproperty:high value}, if $i$ was a non-cutter in $\allocation$ who was formerly allocated her favorite complete piece, then 
	$\gain{\allocation'}{i}  \geq 0$ and hence  $\gain{\allocation'}{i}  \geq - \gain{\allocation}{i}$. 
	Similarly, if $i$ was the cutter in $\allocation$, then by line \ref{line:correction:last} of \correct we have $\gain{\allocation'}{i} \geq 0 \geq -\gain{\allocation}{i} $.
	Therefore the difficulty lies in proving that the  property~\ref{ctproperty:gain bound} holds for non-cutters who did not get their favorite whole piece when $\allocation$ was produced. 
	Suppose that after the permutation output by $\correct$, such a non-cutter $i$ is still allocated the same piece $p_i$. By $\core$ property \ref{cproperty:envy-free}, we again have $\gain{\allocation'}{i}  \geq 0 \geq - \gain{\allocation}{i}$.
	The remaining, and most interesting, case is when such a non-cutter $i$ is allocated
	a piece $p_i'$ 
	and some other agent is allocated $i$'s former (and favorite among the pieces in $\allocation$) piece.
	By Lemma~\ref{lem:secondfav}  we have that it is always the case that $v_i(p_i') \geq v_i(s_i)$, and therefore
	$\gain{\allocation'}{i} \ge v_i(p_i') - v_i(p_i) \geq v_i(s_i) - v_i(p_i) = - \gain{\allocation}{i}$. 
	%
	\qed
\end{proofof}


\bibliographystyle{plainnat}
\bibliography{fairdivrefs}

\end{document}